\documentclass[10pt,twocolumn,twoside]{IEEEtran}
\usepackage{xcolor}

\usepackage{amssymb}
\usepackage{amsmath}
\usepackage{amsthm}

\usepackage{cite}
\usepackage{esvect}
\usepackage[colorlinks,urlcolor=blue,linkcolor=blue,citecolor=blue]{hyperref}
\usepackage{graphicx}
\usepackage{algorithm}
\usepackage{algpseudocode}

\theoremstyle{definition}
\newtheorem{proposition}{Proposition}
\newtheorem{definition}{Definition}

\newtheorem{lemma}{Lemma}
\newtheorem{theorem}{Theorem}
\newtheorem{remark}{Remark}

\newcommand{\Black}[1]{\textcolor{black}{{#1}}}

\newcommand{\changed}[1]{\Black{#1}}

\DeclareMathOperator{\rank}{rank}

\DeclareMathOperator{\blkdiag}{blkdiag}
\def\th{\text{th}}

\def\Rbb{\mathbb{R}}
\def\Zbb{\mathbb{Z}}

\def\Ecal{\mathcal{E}}
\def\Gcal{\mathcal{G}}
\def\Lcal{\mathcal{L}}
\def\Ncal{\mathcal{N}}
\def\Ucal{\mathcal{U}}
\def\Vcal{\mathcal{V}}
\def\Xcal{\mathcal{X}}
\def\Ycal{\mathcal{Y}}

\def\Lfrak{\mathfrak{L}}
\def\Mfrak{\mathfrak{M}}
\def\Rfrak{\mathfrak{R}}
\def\Cfrak{\mathfrak{C}}

\def\ubf{\mathbf{u}}
\def\wbf{\mathbf{w}}
\def\xbf{\mathbf{x}}
\def\ybf{\mathbf{y}}

\def\Kbf{\mathbf{K}}
\def\Phibf{\mathbf{\Phi}}
\def\Psibf{\mathbf{\Psi}}

\def\Phivec{\vv{\Phibf}}
\def\Lambvec{\vv{\Lambda}}

\def\Z{Z_{AB}}
\def\Zi{\Z^\dagger}
\def\Fi{F^\dagger}
\def\Hi{H^\dagger}
\def\IOBlock{\begin{bmatrix} I \\ 0 \end{bmatrix}}
\def\Zb{\Z^\text{blk}}
\def\Zhb{Z_h^\text{blk}}

\def\nPhi{N_\Phi}
\def\nL{N_\Lfrak}
\def\nM{N_\Mfrak}

\title{\LARGE \bf Global Performance Guarantees for Localized Model Predictive Control}

\author{Jing Shuang (Lisa) Li and Carmen Amo Alonso
	\thanks{Corresponding author: J. S. Li (e-mail: {\tt\small jslisali@umich.edu})}
}
\begin{document}

\maketitle

\begin{abstract}

Recent advances in model predictive control (MPC) leverage local communication constraints to produce localized MPC algorithms whose complexities scale independently of total network size. However, no characterization is available regarding global performance, i.e. whether localized MPC (with communication constraints) performs just as well as global MPC (no communication constraints). In this paper, we provide analysis and guarantees on global performance of localized MPC --- in particular, we derive sufficient conditions for optimal global performance in the presence of local communication constraints. We also present an algorithm to determine the communication structure for a given system that will preserve performance while minimizing computational complexity. The effectiveness of the algorithm is verified in simulations, and additional relationships between network properties and performance-preserving communication constraints are characterized. A striking finding is that in a network of 121 coupled pendula, each node only needs to communicate with its immediate neighbors to preserve optimal global performance. Overall, this work offers theoretical understanding on the effect of local communication on global performance, and provides practitioners with the tools necessary to deploy localized model predictive control by establishing a rigorous method of selecting local communication constraints. This work also demonstrates --- surprisingly --- that the inclusion of severe communication constraints need not compromise global performance.

\end{abstract}
\section{INTRODUCTION} \label{sec:introduction}

Distributed control is crucial for the operation of large-scale networks such as power grids and intelligent transport systems. Distributed model predictive control (MPC) is of particular interest, since MPC is one of the most powerful and commonly used control methods. \changed{Many formulations for distributed MPC involve open-loop policies (i.e. optimize directly over states and inputs) \cite{venkat_distributed_2008,zheng_networked_2013,giselsson_accelerated_2013,conte_distributed_2016,jalal_limited-communication_2017,wang_distributed_2015, venkat2005stability,sturz_distributed_2020}. Distributed closed-loop approaches (i.e. optimize over policies) are rarer and typically require strong assumptions, such as the existence of a static structured stabilizing controller \cite{conte_robust_2013} or decoupled subsystems \cite{richards_robust_2007}. However, distributed closed-loop approaches tend to be more amenable for extension to robust formulations --- this motivated the development of the distributed and localized MPC (DLMPC) formulation, introduced in previous work \cite{amoalonso_dlmpc_cdc,amoalonso_dlmpc1_journal,amoalonso_dlmpc2_journal}.} DLMPC is unique among distributed MPC methods in that it computes structured closed-loop policies, can be solved at scale via distributed optimization, and requires no strong assumptions on the system: it may be used on arbitrary linear systems. In the context of system level synthesis, upon which DLMPC is based, DLMPC is also the first work to extend to system level formulation to the distributed online (i.e. predictive) setting. Additionally, DLMPC enjoys minimally conservative feasibility and stability guarantees \cite{amoalonso_dlmpc2_journal}. However, DLMPC requires the inclusion of local communication constraints, whose effects on performance is, thus far, unexplored --- this is the focus of this work.

The key benefits of DLMPC are facilitated by the inclusion of local communication constraints, which are typically encapsulated in a single parameter $d$ (rigorously defined in Section \ref{sec:sls}). The question of how to select this parameter remains unresolved, as two opposing forces come into play: smaller values of $d$ represent stricter communication constraints, which correspond to decreased complexity --- however, overly strict communication constraints may render the problem infeasible, or compromise system performance. In this work, we address this problem by providing a rigorous characterization of the impact of local communication constraints on performance.

The analysis in this work is enabled by the unique formulation of DLMPC with regards to communication constraints. While previous distributed MPC approaches typically 1) focus on iterative distributed solutions of the centralized problem \cite{venkat_distributed_2008, jalal_limited-communication_2017, giselsson_accelerated_2013, sturz_distributed_2020} or 2) reshape performance objectives for distributed optimization \cite{zheng_networked_2013, conte_distributed_2016}, DLMPC does neither. \changed{In DLMPC, we begin with a reparametrization of the standard centralized MPC problem, then introduce an additional constraint of local communication, which manifests as a sparsity constraint on the decision variables. No changes to the performance objective are required. The resulting \textit{localized MPC} problem can be exactly and optimally solved via distributed optimization --- that is, we can apply the alternating direction method of multipliers (ADMM) \cite{boyd2011distributed} to split the problem into subproblems, which are solved in parallel. For algorithm convergence, no assumptions are required on the underlying system graph, which is constructed solely using system matrices (see Definition \ref{defn:graph}). However, when the system graph is sparse (see Section \ref{sec:simulations} for an example), the local communication constraints confer additional scalability upon the distributed problem --- in particular, the size of subproblems do not depend on the total size of the network \cite{Anderson2019, amoalonso_dlmpc_cdc, amoalonso_dlmpc1_journal}}.

In this work, we rigorously analyze the effect of this local constraint by comparing the performance of the localized MPC problem to the standard global MPC problem. We restrict analysis to the linear setting, and focus on cases in which optimal global performance (with respect to any convex objective function with its minimum at the origin) may be obtained with local communication. In other words, the performance of the system is unchanged by the introduction of communication constraints. A striking finding is that in a network of coupled pendula, optimal global performance can be achieved with relatively strict local communication constraints --- in fact, if every subsystem is actuated, then each subsystem only needs to communicate with its immediate neighbors to preserve optimal global performance.

\textbf{Prior work:} For large networked systems, several studies have been conducted on the use of offline (i.e. fixed-policy) controllers with local communication constraints. Local communication can facilitate faster computational speed \cite{Gasparri2020} and convergence \cite{Ballotta2023b}, particularly in the presence of delays \cite{Ballotta2023} --- however, this typically comes at the cost of supoptimal global performance \cite{Jiao2020}. In \cite{Shin2022}, a trade-off between performance and decentralization level (i.e. amount of global communication) is found for a truncated linear quadratic regulator. In system level synthesis, the offline predecessor of DLMPC \cite{Anderson2019}, localization is typically associated with reduced performance of around 10\% relative to the global controller. More generally, for both global and localized control, the topology of the network and actuator placement \cite{Summers2018} plays a role in achievable controller performance \cite{Mousavi2020, Tang2018} and convergence \cite{Baras2008}. In the realm of predictive control, communication constraints are important considerations \cite{Lucia2015}. However, the improved computational speeds offered by local predictive controllers typically come at the cost of suboptimal global performance and lack of stability and convergence guarantees \cite{Bemporad2010}. The novel DLMPC method \cite{amoalonso_dlmpc_cdc, amoalonso_dlmpc1_journal, amoalonso_dlmpc2_journal} overcomes some of these drawbacks, providing both stability and convergence guarantees --- however, thus far, its performance has not been substantially compared to that of the global, full-communication controller.\footnote{Prior work focuses on comparisons between centralized and distributed optimization schemes for localized MPC}. In the few instances that it has, it performed nearly identically to the global controller despite the inclusion of strict communication constraints \cite{Li2021_Layered}, prompting further investigation.

\textbf{Contributions:} This work contains two key contributions. First, we provide a rigorous characterization of how local communication constraints restrict (or preserve) the set of trajectories available under predictive control, and use this to provide guarantees on optimal global performance for localized MPC. Secondly, we provide an exact method for selecting an appropriate locality parameter $d$ for localized MPC. To the best of our knowledge, these are the first results of this kind on local communication constraints; our findings are useful to theoreticians and practitioners alike.

\textbf{Paper structure:} We first provide the appropriate mathematical formulations for global MPC and localized MPC in Section \ref{sec:sls}, then leverage these formulations to provide rigorous characterizations of global and localized MPC performance in Section \ref{sec:locality_analysis} --- this is the main theoretical contribution of the paper. We then use these results to create an algorithm that determines the optimal local communication parameter $d$ in Section \ref{sec:locality_selection}. We run this algorithm on a power-system inspired grid of coupled pendula, and observe the dependence of $d$ on various system parameters in Section \ref{sec:simulations}, and conclude with directions of future work in Section \ref{sec:conclusions}.

\subsection{NOTATION}
Lower-case and upper-case letters such as $x$ and $A$ denote vectors and matrices respectively, although lower-case letters might also be used for scalars or functions (the distinction will be apparent from the context). An arrow above a matrix quantity denotes vectorization, i.e. $\vv{A}$ is the vectorization of $A$. 
Boldface lower case letters such as $\xbf$ denote finite horizon signals. Unless otherwise stated, boldface upper case letters such as $\Kbf$ denote causal (i.e. lower block triangular) finite horizon operators: 
\begin{equation}
    \xbf = \begin{bmatrix} x_0 \\ x_1 \\ \vdots \\ x_T \end{bmatrix},
    \Kbf = \begin{bmatrix} K_{0,0} & & & \\
                           K_{1,1} & K_{1,0} & & \\
                           \vdots & \ddots & \ddots & \\
                           K_{T,T} & \ldots & K_{T,1} & K_{T,0}\end{bmatrix}
\end{equation}
where each $K_{i,j}$ is a matrix. Unless required, dimensions are not stated, and compatible dimension can be assumed. 

For matrix $A$, $(A)_{i,:}$ denotes the $i^\th$ row of $A$, $(A)_{:,j}$ denotes the $j^\th$ column, and $(A)_{i,j}$ denotes the element in the $i^\th$ row and $j^\th$ column. $(A)_{i:,:}$ denotes the rows of $A$ starting from the $i^\th$ row. Calligraphic letters such as $\Ycal$ denote sets, and script letters such as $\Rfrak$ denote a subset of $\Zbb^+$, e.g. $\Rfrak = \{1, \ldots, n \} \subset \Zbb^+$. For matrix $A$, $(A)_{\Rfrak,\Cfrak}$ denotes the submatrix of $A$ composed of the rows and columns specified by $\Rfrak$ and $\Cfrak$, respectively; for vector $x$, $(x)_{\Rfrak}$ denotes the vector composed of the elements specified by $\Rfrak$. 

Bracketed indices denote timestep of the real system --- for example, the system state is $x(\tau)$ at time $\tau$, not to be confused with $x_t$ which denotes the \textit{predicted} state $x$ in $t$ timesteps. Square bracket notation, e.g. $[x]_{i}$, denotes the components of $x$ corresponding to subsystem $i$.

For ease of variable manipulation, we also introduce augmented variables. For any matrix $Z$, the corresponding augmented matrix $Z^\text{blk}$ is defined as a block-diagonal matrix containing $N_x$ copies of $Z$, i.e. $Z^\text{blk} := \blkdiag(Z, \ldots Z)$. For any matrix $Y = Z \Lambda$, the corresponding vectorization can be written as $\vv{Y} = Z^\text{blk} \Lambvec$.

For state $x_0$, the corresponding augmented state $X$ is defined as $X(x_0) := \begin{bmatrix}(x_0)_1I & (x_0)_2I & \ldots & (x_0)_{N_x}I \end{bmatrix}$. For notational simplicity, we write $X$ instead of $X(x_0)$; dependence on $x_0$ is implicit. For any matrix $\Lambda$, $\Lambda x_0 = X \Lambvec$.

\section{LOCALIZED MPC} \label{sec:sls}
We begin with a brief summary of global MPC and localized MPC \cite{amoalonso_dlmpc_cdc, amoalonso_dlmpc1_journal, amoalonso_dlmpc2_journal}.

Consider the standard linear time-invariant discrete-time system:
\begin{equation} \label{eq:dynamics}
    x(t+1) = Ax(t) + Bu(t)
\end{equation}
where $x \in \Rbb^{N_x}$ is the state and $u \in \Rbb^{N_u}$ is the control input. System \eqref{eq:dynamics} can be interpreted as $N$ interconnected subsystems, each equipped with its own sub-controller. 
\changed{\begin{definition} \label{defn:graph}
$\Gcal_{(A,B)}(\Ecal,\Vcal)$ is the \textit{system graph}, which is undirected and unweighted. Assign each state $k \in \{1 \ldots N_x\}$ and actuation $l \in \{1 \ldots N_u \}$ to some subsystem $i \in \{1 \ldots N\}$. Each subsystem $i$ is identified with a vertex $v_{i}\in \Vcal$, and an edge $e_{ij}\in \Ecal$ exists whenever $[A]_{ij}\neq 0$ or $[B]_{ij}\neq 0$ for subsystems $i$ and $j$. 
\end{definition}}

\changed{A natural choice is to assign each state to its own subsystem, and to assign actuation to the states they act upon. Occasionally, two or more states may be assigned to the same subsystem --- for example, when two states represent two different measurements (e.g. phase and frequency) on the same physical location. See Section \ref{sec:simulations} for an example.}

The MPC problem at each timestep $\tau$ is defined as follows:

\begin{subequations} \label{eq:mpc_original}
\begin{align}
\label{eq:mpc_original_obj}
\underset{x_t, u_t, \gamma_t}{\min} \quad & \sum_{t=0}^{T-1}f_{t}(x_{t},u_{t}) + f_{T}(x_{T}) \\
\text{s.t.} \quad & x_0 = x(\tau), \\
& x_{t+1} = Ax_t + Bu_t, \\ 
\label{eq:mpc_original_constr}
& x_T \in \Xcal_T, ~ x_{t} \in \Xcal_t, ~ u_t \in \Ucal_t, \\ \label{eq:mpc_controller}
& u_{t} = \gamma_t(x_{0:t},u_{0:t-1}), ~ t=0,...,T-1
\end{align}
\end{subequations}
where $f_t(\cdot,\cdot)$ and $f_T(\cdot)$ are assumed to be closed, proper, and convex, with the minimum at the origin; $\gamma_t(\cdot)$ is a measurable function of its arguments; and sets $\Xcal_T$, $\Xcal_t$, and $\Ucal$ are assumed to be closed and convex sets containing the origin for all $t$.
\changed{
\begin{remark}
Problem \eqref{eq:mpc_original} is the closed-loop version of the standard MPC problem --- they are mathematically equivalent. Policies $\gamma_t$ are time-varying and capture all possible causal combinations of state $x_t$ and input $u_t$.\footnote{Readers are referred to \cite{amoalonso_dlmpc1_journal} for details.} 
\end{remark}
}

In localized MPC, we impose communication constraints such that each subsystem can only communicate with a other subsystems within its local region. We shall refer to these constraints interchangeably as \textit{local communication constraints} or \textit{locality constraints}. Let $\Ncal(i)$ denote the set of nodes that subsystem $i$ can communicate with them. Then, to enforce locality constraints on \eqref{eq:mpc_original}, we replace constraint \eqref{eq:mpc_controller} with
\begin{equation}
    [u_t]_i = \gamma_{i,t}([x_{0:t}]_{j \in \Ncal(i)}, [u_{0:t-1}]_{j \in \Ncal(i)} )
\end{equation}

Most SLS-based works use the concept of $d$-local neighborhoods, replacing $\Ncal(i)$ with $\Ncal_d(i)$:
\begin{definition}
    $\Ncal_d(i)$ denotes the \textit{$d$-local neighborhood of subsystem $i$}. Subsystem $j \in \Ncal_d(i)$ if there exists a path of $d$ or less edges between subsystems $i$ and $j$ in $\Gcal_{(A,B)}(\Ecal,\Vcal)$.
\end{definition}

The inclusion of local communication constraints renders problem \eqref{eq:mpc_original} difficult to solve. To ameliorate this, we can apply ideas from the system level synthesis framework \cite{Anderson2019} to reparametrize the problem \cite{amoalonso_dlmpc1_journal} --- we now introduce relevant ideas from this framework. 

First, we define $\hat{A}:=\blkdiag(A,...,A)$, $\hat {B}:=\blkdiag(B,...,B)$, and $Z$, the block-downshift matrix --- with identity matrices along the first block sub-diagonal and zeros elsewhere. The behavior of system \eqref{eq:dynamics} over time horizon $t = 0 \ldots T$ can be written in terms of \textit{closed-loop maps} in signal domain:
\begin{subequations} \label{eq:closed_loop_maps}
    \begin{equation}
        \xbf = (I - Z(\hat{A} + \hat{B}\Kbf))^{-1}\wbf =: \Psibf_x \wbf
    \end{equation}
    \begin{equation}
        \ubf = \Kbf\Psibf_x\wbf =: \Psibf_u \wbf
    \end{equation}
\end{subequations}
where $\wbf$ represents driving noise; in our case, where no driving noise is present, it consists only of the initial condition, i.e. $\wbf = \begin{bmatrix} x_0^\top & 0 & \ldots & 0 \end{bmatrix}^\top$. The key idea is to use closed-loop maps $\Psibf_x$, $\Psibf_u$ as optimization variables, subject to appropriate constraints.

\begin{theorem} \label{thm:sls} (Theorem 2.1 of \cite{Anderson2019}) For system \eqref{eq:dynamics} with causal state feedback control policy $\ubf = \Kbf \xbf$,
\begin{enumerate}
    \item The affine subspace of causal closed-loop maps
    \begin{equation} \label{eq:sls_theorem}
        \Z \begin{bmatrix}
            \Psibf_x \\ \Psibf_u
        \end{bmatrix}
        = I
    \end{equation}
    parametrizes all achievable closed-loop maps $\Psibf_x$, $\Psibf_u$. Here, $\Z := \begin{bmatrix} I - Z\hat{A} & -Z\hat{B} \end{bmatrix}$
    \item For any closed-loop maps $\Psibf_x$, $\Psibf_u$ obeying \eqref{eq:sls_theorem}, the controller $\Kbf = \Psibf_u\Psibf_x^{-1}$ achieves the desired closed-loop response as per \eqref{eq:closed_loop_maps}.
\end{enumerate}    
\end{theorem}

Theorem \ref{thm:sls} allows us to reformulate a control problem over state and input signals into an equivalent problem over closed-loop maps $\Psibf_x$, $\Psibf_u$. In the case with no driving noise, only the first block column of $\Psibf_x$, $\Psibf_u$ need to be computed due to the zeros in $\wbf$. We can rewrite \eqref{eq:closed_loop_maps} as $\xbf = \Phibf_x x_0$ and $\ubf = \Phibf_u x_0$, where $\Phibf_x$ and $\Phibf_u$ correspond to the first block columns of $\Psibf_x$ and $\Psibf_u$, respectively. Then, we can rewrite \eqref{eq:sls_theorem} as 
\begin{equation} \label{eq:dynamics_constr_xu}
    \Z\begin{bmatrix}\Phibf_x \\ \Phibf_u \end{bmatrix} = \IOBlock
\end{equation}
Notice that for any $\Phibf_x$ satisfying constraint \eqref{eq:dynamics_constr_xu}, $(\Phibf_x)_{1:N_x,:} = I$. This is due to the structure of $\Z$. Also, constraint \eqref{eq:dynamics_constr_xu} is always feasible, with solution space of dimension $N_uT$. To see this, notice that $\rank(\Z) = \rank\begin{bmatrix} \Z & \IOBlock \end{bmatrix}$ always holds, since $\Z$ has full row rank due to the identity blocks on its diagonal; apply the Rouch\'e-Capelli theorem \cite{shafarevich2012linear} to get the desired result.

We now apply this closed-loop parametrization to \eqref{eq:mpc_original}, as is done in \cite{amoalonso_dlmpc1_journal}:

\begin{subequations} \label{eq:mpc_sls}
\begin{align}
\label{eq:lmpc_obj}
\underset{\Phibf_x, \Phibf_u}{\min} \quad & f(\Phibf_x x_0, \Phibf_u x_0) \\ 
\text{s.t.} \quad & x_0 = x(\tau), \\ 
\label{eq:lmpc_achiev_constraint}
& \Z\begin{bmatrix}\Phibf_x \\ \Phibf_u \end{bmatrix} = \IOBlock, \\
\label{eq:lmpc_constr}
& \Phibf_x x_0 \in \Xcal, ~ \Phibf_u x_0 \in \Ucal
\end{align}
\end{subequations}
Objective $f$ is defined such that \eqref{eq:mpc_original_obj} and \eqref{eq:lmpc_obj} are equivalent; similarly, constraint sets $\Xcal$, $\Ucal$ are defined such that \eqref{eq:mpc_original_constr} and \eqref{eq:lmpc_constr} are equivalent. 
\changed{
\begin{remark}
    Problem \eqref{eq:mpc_sls} is mathematically equivalent to problem \eqref{eq:mpc_original}. Both are alternative --- but equivalent --- formulations of the standard MPC problem.
\end{remark}
}

In \eqref{eq:mpc_sls}, $\Phibf_x$ and $\Phibf_u$ not only represent the closed-loop maps of the system, but also the communication structure of the system. For instance, if $[\Phibf_x]_{i,j} = 0$ and $[\Phibf_u]_{i,j} = 0 \quad \forall i \neq j $, then node $i$ requires no knowledge of $(x_0)_{j \neq i}$ --- consequently, no communication is required between nodes $i$ and $j$ for all $j \neq i$. The relationship between closed-loop maps and communication constraints are further detailed in \cite{amoalonso_dlmpc1_journal}. Thus, to incorporate local communication into this formulation, we introduce an additional constraint:
\begin{equation} \label{eq:loc_constraint_xu}
\Phibf_x \in \Lcal_x, ~ \Phibf_u \in \Lcal_u 
\end{equation}
where $\Lcal_x$ and $\Lcal_u$ are sets with some prescribed sparsity pattern that is compatible with the desired local communication constraints.

For simplicity, we define decision variable $\Phibf := \begin{bmatrix} \Phibf_x \\ \Phibf_u \end{bmatrix}$, which has $\nPhi := N_x(T+1) + N_uT$ rows. We also rewrite locality constraints \eqref{eq:loc_constraint_xu} as $\Phibf \in \Lcal$.

From here on, we shall use \textit{global MPC} to refer to \eqref{eq:mpc_original}, or equivalently, \eqref{eq:mpc_sls}. We shall use \textit{localized MPC} to refer to \eqref{eq:mpc_sls} with constraint \eqref{eq:loc_constraint_xu}. We remark that for appropriately chosen locality constraints, localized MPC confers substantial scalability benefits \cite{amoalonso_dlmpc1_journal}. \changed{In particular, we can use ADMM \cite{boyd2011distributed} to split the localized MPC problem into subproblems, which are solved in parallel, i.e. each subsystem solves its own subproblem. When the system graph is sparse (see Section \ref{sec:simulations} for an example), the size of subproblems do not depend on the total size of the network \cite{Anderson2019, amoalonso_dlmpc_cdc, amoalonso_dlmpc1_journal}.}

\section{GLOBAL PERFORMANCE OF LOCALIZED MPC} \label{sec:locality_analysis}

In this section, we analyze the effect of locality constraints $\Phibf \in \Lcal$ on MPC performance. We are especially interested in scenarios where localized MPC achieves \textit{optimal global performance}, i.e. $f^* = f^*_\Lcal$, where $f^*$ and $f^*_\Lcal$ are the solutions to the global MPC problem and localized MPC problem, respectively, for some state $x_0$. 

First, we must analyze the space of available trajectories from state $x_0$ for both global and localized MPC. We denote an available trajectory $\ybf := \begin{bmatrix} \xbf_{1:T} \\ \ubf \end{bmatrix}$.

\begin{definition} \textit{Trajectory set} $\Ycal(x_0)$ denotes the set of available trajectories from state $x_0$ under dynamics \eqref{eq:dynamics}:
\begin{equation*}
    \Ycal(x_0) := \{ \ybf: \exists \Phibf \text{ s.t. } \Z\Phibf = \IOBlock, \ybf=(\Phibf)_{N_x+1:,:} x_0\}
\end{equation*}
\textit{Localized trajectory set} $\Ycal_\Lcal(x_0)$ denotes the set of available trajectories from state $x_0$ under dynamics \eqref{eq:dynamics} and locality constraint $\Phibf \in \Lcal$:
\begin{equation*}
\begin{aligned}
    \Ycal_\Lcal(x_0) := \{ \ybf: \exists \Phibf \text{ s.t. } \Z\Phibf = \IOBlock, \\ \Phibf \in \Lcal, \quad \ybf=(\Phibf)_{N_x+1:,:} x_0\}
\end{aligned}
\end{equation*}
\end{definition}

\begin{proposition} \label{prop:global_optimality}
\textit{(Optimal global performance)} For state $x_0$, if the local communication constraint set $\Lcal$ is chosen such that $\Ycal(x_0) = \Ycal_\Lcal(x_0)$, then the localized MPC problem will attain optimal global performance.
\end{proposition}
\begin{proof}
Global MPC problem \eqref{eq:mpc_original} can be written as
\begin{subequations} \label{eq:lmpctraj}
\begin{align}
\underset{\xbf, \ubf}{\min} \quad & f(\xbf, \ubf) \\
\text{s.t.} \quad & x_0 = x(\tau), ~ \xbf \in \Xcal, ~ \ubf \in \Ucal, \\
& \ybf := \begin{bmatrix} \xbf_{1:T} \\ \ubf \end{bmatrix} \in \Ycal(x_0) \label{eq:traj_constraint}
\end{align}
\end{subequations}
The localized MPC problem can also be written in this form, 
by replacing $\Ycal(x_0)$ in constraint \eqref{eq:traj_constraint} with $\Ycal_\Lcal(x_0)$. Thus, if $\Ycal(x_0) = \Ycal_\Lcal(x_0)$, the two problems are equivalent and will have the same optimal values.
\end{proof}

\textit{Remark:} This is a sufficient but not necessary condition for optimal global performance. Even if this condition is not satisfied, i.e. $\Ycal_\Lcal(x_0) \subset \Ycal(x_0)$, the optimal global trajectory may be contained within $\Ycal_\Lcal(x_0)$. However, this is dependent on objective $f$. Our analysis focuses on stricter conditions which guarantee optimal global performance for \textit{any} objective function.

We now explore cases in which $\Ycal(x_0) = \Ycal_\Lcal(x_0)$, i.e. the localized MPC problem attains optimal global performance. Localized trajectory set $\Ycal_\Lcal(x_0)$ is shaped by the dynamics and locality constraints:
\begin{subequations}
\begin{align}
    \label{eq:constr1}
    & \Z\Phibf = \IOBlock \\
    \label{eq:constr2}
    & \Phibf \in \Lcal
\end{align}
\end{subequations}

To obtain a closed-form solution for $\Ycal_\Lcal(x_0)$, we will parameterize these constraints. Two equivalent formulations are available. The dynamics-first formulation parameterizes constraint \eqref{eq:constr1}, then \eqref{eq:constr2}, and the locality-formulation parameterizes the constraints in the opposite order. The dynamics-first formulation clearly shows how local communication constraints affect the trajectory space; the locality-first formulation is less clear in this regard, but can be implemented in code with lower computational complexity than the dynamics-first formulation. We now derive each formulation.

\subsection{DYNAMICS-FIRST FORMULATION} \label{sec:formulation1}

We first parameterize \eqref{eq:constr1}, which gives a closed-form expression for trajectory set $\Ycal(x_0)$.

\begin{lemma} \label{lemm:yset_formulation1} \textit{(Image space representation of trajectory set)} 
\begin{enumerate}
    \item The trajectory set from state $x_0$ is described by:
\begin{equation*}
    \Ycal(x_0) = \{ \ybf : \ybf = Z_p x_0 + Z_h X \lambda, \quad \lambda \in \Rbb^{\nPhi} \}
\end{equation*}
\begin{equation*}
\begin{aligned}
\text{where} \quad Z_p & := (\Zi)_{N_x+1:,:} \IOBlock \\ 
\text{and} \quad Z_h & := (I - \Zi\Z)_{N_x+1:,:}
\end{aligned}
\end{equation*}

and the size of the trajectory set is
\begin{equation*}
\dim(\Ycal(x_0)) = \rank(Z_hX)    
\end{equation*}
    \item If $x_0$ has at least one nonzero value, then 
    \begin{equation*}
    \dim(\Ycal(x_0)) = N_uT    
    \end{equation*}
\end{enumerate}
\end{lemma}
\begin{proof}
As previously shown, \eqref{eq:constr1} always has solutions. We can parameterize the space of solutions $\Phibf$ as:
\begin{equation} \label{eq:phi_constraint}
    \Phibf = \Zi \IOBlock + (I - \Zi\Z) \Lambda
\end{equation}
where $\Lambda$ is a free variable with the same dimensions as $\Phibf$. Recall that $\Phibf_{1:N_x,:} = I$ always holds --- thus, we can omit the first $N_x$ rows of \eqref{eq:phi_constraint}. Define $\Phibf_2 := (\Phibf)_{N_x+1:,:}$ and consider
\begin{equation} \label{eq:phi2_constraint}
    \Phibf_2 = Z_p + Z_h\Lambda
\end{equation}
Combining \eqref{eq:phi2_constraint} and the definition of $\ybf$, we have
\begin{equation}
    \ybf = \Phibf_2x_0 = Z_px_0 + Z_h\Lambda x_0
\end{equation}
Making use of augmented state $X$, rewrite this as
\begin{equation} \label{eq:trajectory_expr}
    \ybf = \Phibf_2x_0 = Z_px_0 + Z_h X \Lambvec
\end{equation}
This gives the desired expression for $\Ycal(x_0)$ and its size.

To prove 2), notice that if $x_0$ has at least one nonzero value, then $\rank(Z_hX) = \rank(Z_h)$ due to the structure of $X$. All that is left is to show $\rank(Z_h) = N_uT$. First, note that $\rank(\Z) = N_x(T+1)$ due to the identity blocks on the diagonal of $\Z$. It follows that $\rank(\Zi) = \rank(\Zi \Z) = N_x(T+1) $. Thus, $\rank(I-\Zi \Z) = \nPhi - N_x(T+1) = N_uT$. Recall that $Z_h$ is simply $I-\Zi \Z$ with the first $N_x$ rows removed; this does not result in decreased rank, since all these rows are zero (recall that $\Phibf_{1:N_x,:}$ is always equal to $I$). Thus, $\rank(Z_h) = \rank(I-\Zi \Z) = N_uT$.
\end{proof}

To write the closed form of localized trajectory set $\Ycal_\Lcal(x_0)$, we require some definitions:

\begin{definition} \label{defn:loc_constrained_indices} \textit{Constrained vector indices} $\Lfrak$ denote the set of indices of $\Phibf_2 := (\Phibf)_{N_x+1:,:}$ that are constrained to be zero by the locality constraint \eqref{eq:constr2}, i.e.
\begin{equation*}
    (\vv{\Phibf_2})_{\Lfrak} = 0 \Leftrightarrow \Phibf \in \Lcal
\end{equation*}
Let $\nL$ be the cardinality of $\Lfrak$.
\end{definition}

 We now parameterize \eqref{eq:constr2} and combine this with Lemma \ref{lemm:yset_formulation1}, which gives a closed-form expression for localized trajectory set $\Ycal_\Lcal(x_0)$.
 
\begin{lemma} \label{lemm:ylset_formulation1}
\textit{(Image space representation of localized trajectory set)} 
Assume there exists some $\Phibf$ that satisfies constraints \eqref{eq:constr1} and \eqref{eq:constr2}. Then, the localized trajectory set from state $x_0$ is described by the following:
\begin{equation*}
\begin{aligned}
    \Ycal_\Lcal(x_0) = \{ \ybf: \quad & \ybf = Z_p x_0 + Z_h X \Fi g \quad + \\ 
    & Z_h X (I - \Fi F)\mu, \quad \mu \in \Rbb^{\nL} \} 
\end{aligned}
\end{equation*}
\begin{equation*}
\text{where} \quad F := (\Zhb)_{\Lfrak,:} \quad \text{and} \quad g := -(\vv{Z_p})_{\Lfrak}   
\end{equation*}
and the size of the localized trajectory set is 
\begin{equation*}
\dim(\Ycal_\Lcal(x_0)) = \rank(Z_h X (I-\Fi F))    
\end{equation*}
\end{lemma}

\begin{proof}
Using the augmented matrix of $Z_h$, we can write the vectorization of \eqref{eq:phi2_constraint} as
\begin{equation}
    \vv{\Phibf_2} = \vv{Z_p} + \Zhb\Lambvec
\end{equation}
where $\Lambvec$ is a free variable. Incorporate locality constraint \eqref{eq:constr2} using the constrained vector indices:
\begin{equation}
    (\vv{Z_p} + \Zhb\Lambvec)_\Lfrak = 0
\end{equation}
This is equivalent to $(\vv{Z_p})_\Lfrak + (\Zhb)_{\Lfrak,:} \Lambvec = 0$, or $F\Lambvec = g$. We can parameterize this constraint as
\begin{equation}
    \Lambvec = \Fi g + (I-\Fi F)\mu
\end{equation}
where $\mu$ is a free variable. Plugging this into \eqref{eq:trajectory_expr} gives the desired expression for $\Ycal_\Lcal(x_0)$ and its size. We remark that there is no need to consider $\Phibf_1 = I$ in relation to the locality constraints, since the diagonal sparsity pattern of the identity matrix (which corresponds to self-communication, i.e. node $i$ `communicating' to itself) satisfies any local communication constraint.
\end{proof}

\begin{theorem} \label{thm:equal_trajsets_1}
\textit{(Optimal global performance)} If $x_0$ has at least one nonzero value, then localized MPC attains optimal global performance if
\begin{enumerate}
    \item there exists some $\Phibf$ that satisfies constraints \eqref{eq:constr1} and \eqref{eq:constr2}, and
    \item $\rank(Z_h X (I-\Fi F)) = N_uT$
\end{enumerate}
\end{theorem}
\begin{proof}
    By definition, $\Ycal_\Lcal(x_0) \subseteq \Ycal(x_0)$. Equality is achieved if and only if the two sets are of equal size. Applying Lemmas \ref{lemm:yset_formulation1} and \ref{lemm:ylset_formulation1} shows that the conditions of the theorem are necessary and sufficient for $\Ycal_\Lcal(x_0)$ and $\Ycal(x_0)$ to be equal. Then, apply Proposition \ref{prop:global_optimality} for the desired result.
\end{proof}

This theorem provides a criterion to assess how the inclusion of locality constraints affects the trajectory set. It also allows us to gain some intuition on the effect of these constraints, which are represented by the matrix $F$. If no locality constraints are included, then $F$ has rank 0; in this case, $\rank(Z_h X (I-\Fi F)) = \rank(Z_h X) = N_uT$, via Lemma \ref{lemm:yset_formulation1}. The rank of $F$ increases as the number of locality constraints increases; this results in decreased rank for $I-\Fi F$, and possibly also decreased rank for $Z_h X (I-\Fi F)$. However, due to the repetitive structure of $Z_h X$, this is not always the case --- it is possible to have locality constraints that do not lead to decreased rank for $Z_h X (I-\Fi F)$. We provide a detailed numerical example of this in Section \ref{sec:numerical_example}.

Unfortunately, checking the conditions of Theorem \ref{thm:equal_trajsets_1} is computationally expensive. In particular, we must assemble and compute the rank of matrix $Z_h X (I-\Fi F)$. The complexity of this operation is dependent on the size of the matrix, which increases as $\Phibf$ becomes more sparse (as enforced by locality constraints). This is a problem, since it is generally preferable to use very sparse $\Phibf$, as previously mentioned --- this would correspond to an extremely large matrix $Z_h X (I-\Fi F)$, which is time-consuming to compute with. Ideally, sparser $\Phibf$ should instead correspond to \textit{lower} complexity; this is the motivation for the next formulation.

\subsection{LOCALITY-FIRST FORMULATION} \label{sec:formulation2}
We first parameterize \eqref{eq:constr2}, then \eqref{eq:constr1}. This directly gives a closed-form expression for localized trajectory set $\Ycal_\Lcal(x_0)$. First, some definitions:

\begin{definition}
\textit{Support vector indices} $\Mfrak$ denote the set of indices of $\Phivec$ such that $(\Phivec)_\Mfrak \neq 0$ is compatible with locality constraint \eqref{eq:constr2}. Let $\nM$ be the cardinality of $\Mfrak$.
\end{definition}

\textit{Remark:} this is complementary to Definition \ref{defn:loc_constrained_indices}. Instead of looking at which indices are constrained to be zero, we now look at which indices are allowed to be nonzero. A subtlety is that this definition considers the entirety of $\Phibf$, while Definition \ref{defn:loc_constrained_indices} omits the first $N_x$ rows of $\Phibf$.

\begin{lemma} \label{lemm:ylset_formulation2}
\textit{(Image space representation of localized trajectory set)} Assume there exists some $\Phibf$ that satisfies constraints \eqref{eq:constr1} and \eqref{eq:constr2}. Then, the localized trajectory set from state $x_0$ is described by the following:
\begin{equation*}
\begin{aligned}
\Ycal_\Lcal(x_0) = \{ \ybf: \quad & \ybf = (X_2)_{:,\Mfrak} \Hi k \quad + \\ & (X_2)_{:,\Mfrak}(I - \Hi H) \gamma, \quad \gamma \in \Rbb^{\nM} \}
\end{aligned}
\end{equation*}

\begin{equation*}
\begin{aligned}
\text{where} \quad X_2 & := (X)_{N_x+1:,:} \\
\text{and} \quad H & := (\Zb)_{:,\Mfrak} \quad \text{and} \quad k := \vv{\IOBlock}  
\end{aligned}
\end{equation*}
and the size of the localized trajectory set is
\begin{equation*}
\dim(\Ycal_\Lcal(x_0)) = \rank((X_2)_{:,\Mfrak}(I - \Hi H))    
\end{equation*}
\end{lemma}

\begin{proof}
Future trajectory $\ybf$ can be written as 
\begin{equation} \label{eq:traj_sparse}
\ybf = X_2 \Phivec = (X_2)_{:,\Mfrak} (\Phivec)_\Mfrak
\end{equation}
where the first equality arises from the definitions of $\ybf$ and $X$, and the second equality arises from the fact that zeros in $\Phivec$ do not contribute to $\ybf$; thus, we only need to consider nonzero values $(\Phivec)_\Mfrak$.

Using the augmented matrix of $\Z$, constraint \eqref{eq:constr1} can be rewritten as $\Zb\Phivec = k$. Nonzero values $(\Phivec)_\Mfrak$ must obey 
\begin{equation} \label{eq:locality_feasibility_2}
    H (\Phivec)_\Mfrak = k
\end{equation}

Constraint \eqref{eq:locality_feasibility_2} is feasible exactly when constraints \eqref{eq:constr1} and \eqref{eq:constr2} are feasible. By assumption, solutions exist, so we can parameterize the solution space as
\begin{equation} \label{eq:soln_space_2}
    (\Phivec)_\Mfrak = \Hi k + (I - \Hi H) \gamma
\end{equation}
where $\gamma$ is a free variable. Substituting \eqref{eq:soln_space_2} into \eqref{eq:traj_sparse} gives the desired expression for $\Ycal_\Lcal(x_0)$ and its size.
\end{proof}

\begin{theorem} \label{thm:equal_trajsets_2}
    \textit{(Optimal global performance)} If $x_0$ has at least one nonzero value, then localized MPC attains optimal global performance if
    \begin{enumerate}
    \item there exists some $\Phibf$ that satisfies constraints \eqref{eq:constr1} and \eqref{eq:constr2}, and
    \item $\rank((X_2)_{:,\Mfrak}(I - \Hi H)) = N_uT$
    \end{enumerate}
\end{theorem}
\begin{proof}
    Similar to Theorem \ref{thm:equal_trajsets_1}; instead of applying Lemma \ref{lemm:ylset_formulation1}, apply Lemma \ref{lemm:ylset_formulation2}.
\end{proof}
    
To check the conditions of Theorem \ref{thm:equal_trajsets_2}, we must assemble and compute the rank of matrix $(X_2)_{:,\Mfrak}(I - \Hi H)$. The complexity of this operation is dependent on the size of this matrix, which decreases as $\Phibf$ becomes more sparse (as enforced by locality constraints). This is beneficial, since it is preferable to use very sparse $\Phibf$, which corresponds to a small matrix $(X_2)_{:,\Mfrak}(I - \Hi H)$ that is easy to compute with. This is in contrast with the previous formulation, in which sparser $\Phibf$ corresponded to increased complexity. However, from a theoretical standpoint, this formulation does not provide any intuition on the relationship between the trajectory set $\Ycal(x_0)$ and the localized trajectory set $\Ycal_\Lcal(x_0)$. 

For completeness, we now use the locality-first formulation to provide a closed-form expression of $\Ycal(x_0)$. The resulting expression is equivalent to --- though decidedly more convoluted than --- the expression in Lemma \ref{lemm:yset_formulation1}:

\begin{lemma} \label{lemm:yset_formulation2} \textit{(Image space representation of trajectory set)} The trajectory set from state $x_0$ is described by:
\begin{equation*}
    \begin{aligned}
    \Ycal(x_0) = \{ \ybf: \quad & \ybf = X_2\Zb k \quad + \\ & X_2 (I - \Z^{\text{blk} \dagger} \Zb) \gamma, \quad \gamma \in \Rbb^{\nPhi}\}
    \end{aligned}
\end{equation*} 
\end{lemma}
\begin{proof}
In the absence of locality constraints, $\Mfrak$ includes all indices of $\Phivec$ since all entries are allowed to be nonzero. Here, $\nM = \nPhi$, $(X_2)_{:,\Mfrak} = X_2$,  $(\Phivec)_{:,\Mfrak} = \Phivec$, and $H = \Zb$. Substitute these into the expression in Lemma \ref{lemm:ylset_formulation2} to obtain the desired result.
\end{proof}

\textit{Remark:} By definition, $X_2\Zb k = Z_p x_0$ and $X_2 (I - \Z^{\text{blk} \dagger} \Zb) = Z_h X$. Substituting these quantities into Lemma \ref{lemm:yset_formulation2} recovers Lemma \ref{lemm:yset_formulation1}.

\subsection{NUMERICAL EXAMPLE} \label{sec:numerical_example}
To provide some intuition on the results from the previous subsections, we present a simple numerical example. We work with a system of three nodes in a chain interconnection, as shown in Fig. \ref{fig:example_system}. The system matrices are:
\begin{equation}
A = \begin{bmatrix} 1 & 2 & 0 \\ 3 & 4 & 5 \\ 0 & 6 & 7 \end{bmatrix}, B = \begin{bmatrix} 1 & 0 \\ 0 & 0 \\ 0 & 1\end{bmatrix}
\end{equation}

\begin{figure}
	\centering
	\includegraphics[width=0.6\columnwidth]{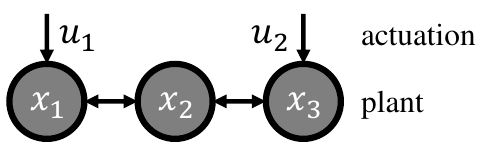}
	\caption{Example system with three nodes, two of which are actuated.}
	\label{fig:example_system}
\end{figure}

We set initial state $x_0 = \begin{bmatrix} 1 & 1 & 1 \end{bmatrix}^\top$, and choose a predictive horizon size of $T=1$. Here, $\nPhi = N_x(T+1) + N_uT = 8$. We choose locality constraints $\Lcal$ such that each node may only communicate with its immediate neighbors; node 1 with node 2, node 2 with both nodes 1 and 3, and node 3 with node 2. Then, locality constraint \eqref{eq:constr2} is equivalent to
\begin{equation}
\Phibf = \begin{bmatrix} \ast & \ast & 0 \\ \ast & \ast & \ast \\ 0 & \ast & \ast \\ \ast & \ast & 0 \\ \ast & \ast & \ast \\ 0 & \ast & \ast \\ \ast & \ast & 0 \\ 0 & \ast & \ast   
\end{bmatrix}
\end{equation}
where $\ast$ indicate values that are allowed to be nonzero. The support vector indices are $\Mfrak = \{ 1,2,4,5,7,9-16, 18, 19, 21, 22, 24\}$, and the constrained vector indices are $\Lfrak = \{3, 5, 11, 14\}$ (recall that indices in $\Lfrak$ do not include the first $N_x$ rows of $\Phibf$). We confirm that there exists some $\Phibf$ that satisfies both dynamics constraints \eqref{eq:constr1} and locality constraints \eqref{eq:constr2} by checking that constraint \eqref{eq:locality_feasibility_2} is feasible.

We start with dynamics-first formulation. In our case,
\begin{equation}
    Z_h = \begin{bmatrix} 0_{5 \times 3} & c_1 & 0_{5 \times 1} & c_2 & c_1 & c_2 \end{bmatrix}
\end{equation}
where $c_1$ and $c_2$ are defined as
\begin{equation}
    c_1 := \frac{1}{2} \begin{bmatrix} 1 \\ 0 \\ 0 \\ 1 \\ 0\end{bmatrix}, c_2 := \frac{1}{2} \begin{bmatrix} 0 \\ 0 \\ 1 \\ 0 \\ 1 \end{bmatrix}
\end{equation}
$Z_h$ has a rank of 2. Then, $Z_h X = \begin{bmatrix} Z_h & Z_h & Z_h \end{bmatrix}$, also has a rank of 2. Per Lemma \ref{lemm:yset_formulation1}, this is the size of the trajectory set $\Ycal(x_0)$, which is exactly equal to $N_uT$, as expected. We write out $Z_h X$ and $Z_h X (I - \Fi F)$ in full in equations \eqref{eq:example_method1_zhx} and \eqref{eq:example_method1_zhxiff}. Boxed columns represent columns zeroed out as a result of locality constraints, i.e. if we replace the boxed columns in $Z_hX$ with zeros, we obtain $Z_h X (I - \Fi F)$. In our example, $Z_h X (I - \Fi F)$ also has a rank of 2; by Theorem \ref{thm:equal_trajsets_1}, the local trajectory set is equal to the trajectory set, and by Theorem \ref{thm:equal_trajsets_1}, the localized MPC problem attains optimal global performance.

\begin{table*}
\begin{align} \label{eq:example_method1_zhx}
& Z_h X = 
\begin{bmatrix}
0_{5 \times 3} & c_1 & 0_{5 \times 1} & \boxed{c_2} & c_1 & \boxed{c_2} & 0_{5 \times 3} & c_1 & 0_{5 \times 1} & c_2 & c_1 & c_2 & 0_{5 \times 3} & \boxed{c_1} & 0_{5 \times 1} & c_2 & \boxed{c_1} & c_2
\end{bmatrix} \\ \label{eq:example_method1_zhxiff}
& Z_h X (I - \Fi F) = 
\begin{bmatrix}
0_{5 \times 3} & c_1 & 0_{5 \times 2} & c_1 & 0_{5 \times 4} & c_1 & 0_{5 \times 1} & c_2 & c_1 & c_2 & 0_{5 \times 5} & c_2 & 0_{5 \times 1} & c_2
\end{bmatrix} \\ \label{eq:example_method2}
& (X_2)_{:,\Mfrak}(I - \Hi H) =
\begin{bmatrix}
0_{5 \times 2} & c_1 & 0_{5 \times 1} & c_1 & 0_{5 \times 3} & c_1 & 0_{5 \times 1} & c_2 & c_1 & c_2 & 0_{5 \times 3} & c_2 & c_2
\end{bmatrix}
\end{align}
\end{table*}

Two observations are in order. First, we notice that the rank of $Z_hX$ (= 2) is low compared to the number of nonzero columns (= 12), especially when $x_0$ is dense. Additionally, the structure of $Z_hX$ is highly repetitive; the only two linearly independent columns are $c_1$ and $c_2$, and each appears 6 times in $Z_hX$. Furthermore, the specific values of $x_0$ do not affect the rank of these matrices --- only the placement of nonzeros and zeros in $x_0$ matters.

Second, notice that post-multiplying $Z_h X$ by $(I-\Fi F)$ effectively zeros out columns of $Z_h X$. However, due to the repetitive structure of $Z_h X$, this does not result in decreased rank for $Z_hX (I - \Fi F)$. In fact, it is difficult to find a feasible locality constraint that results in decreased rank. This observation is corroborated by simulations in Section \ref{sec:simulations}, in which we find that locality constraints that are feasible also typically preserve global performance. For more complex systems and larger predictive horizons, post-multiplication of $Z_h X$ by $(I-\Fi F)$ no longer cleanly corresponds to zeroing out columns, but similar intuition applies.

We now apply the locality-first formulation. To check the conditions of Theorem \ref{thm:equal_trajsets_2}, we must construct the matrix $(X_2)_{:,\Mfrak}(I - \Hi H)$ and check its rank. This matrix is written out in equation \eqref{eq:example_method2}. As expected, the rank of this matrix is also equal to two. Additionally, notice that $(X_2)_{:,\Mfrak}(I - \Hi H)$ contains the same nonzero columns as $Z_h X (I - \Fi F)$: $c_1$ and $c_2$ are each repeated four times, in slightly different orders. This is unsurprising, as the two formulations are equivalent.

\section{ALGORITHMIC IMPLEMENTATION OF OPTIMAL LOCALITY SELECTION} \label{sec:locality_selection}

Leveraging the results of the previous section, we introduce an algorithm that selects the appropriate locality constraints for localized MPC. For simplicity, we restrict ourselves to locality constraints corresponding to $d$-local neighborhoods, though we remark that Subroutine \ref{alg:construct_local_matrix} is applicable to arbitrary communication structures.

The localized MPC problem can be solved via distributed optimization techniques; the resulting distributed and localized MPC problem enjoys complexity that scales with locality parameter $d$, as opposed to network size $N$ \cite{amoalonso_dlmpc1_journal}. Thus, when possible, it is preferable to use small values of $d$ to minimize computational complexity. For a given system and predictive horizon length, Algorithm \ref{alg:main} will return the \textit{optimal locality size} $d$ --- the smallest value of $d$ that attains optimal global performance.

As previously described, the specific values of $x_0$ do not matter --- only the placement of nonzeros and zeros in $x_0$ matters. We will restrict ourselves to considering dense values of $x_0$. For simplicity, our algorithm will work with the vector of ones as $x_0$ --- the resulting performance guarantees hold for \textit{any} dense $x_0$.

To check if a given locality constraint preserves global performance, we must check the two conditions of Theorem \ref{thm:equal_trajsets_2}. First, we must check whether there exists some $\Phibf$ that satisfies both dynamics and locality constraints; this is equivalent to checking whether \eqref{eq:locality_feasibility_2} is feasible. We propose to check whether

\begin{equation} \label{eq:feasibility_check}
\| H(\Hi k) - k \|_\infty \leq \epsilon
\end{equation}
for some tolerance $\epsilon$. Condition \eqref{eq:feasibility_check} can be distributedly computed due to the block-diagonal structure of $H$. Define partitions $[H]_i$ such that $H = \blkdiag([H]_1, [H]_2 \ldots [H]_{N})$\footnote{$H$ should have $N_x$ blocks, where $N_x$ is the number of states. Since $N \leq N_x$ and one subsystem may contain more than one state, we are able to partition $H$ into $N$ blocks as well.}. Then, \eqref{eq:feasibility_check} is equivalent to
\begin{equation} \label{eq:feasibility_check_local}
\| [H]_i([H]_i^\dagger [k]_i) - [k]_i \|_\infty \leq \epsilon \quad \forall i
\end{equation}

If this condition is satisfied, then it remains to check the second condition of Theorem \ref{thm:equal_trajsets_2}. To so, we must construct matrix $J := (X_2)_{:,\Mfrak}(I - \Hi H)$ and check its rank. Notice that $J$ can be partitioned into submatrices $J_i$, i.e. $J := \begin{bmatrix} J_1 & J_2 & \ldots & J_N \end{bmatrix}$, where each block $J_i$ can be constructed using only information from subsystem $i$, i.e. $[H]_i$, $[k]_i$, etc. Thus, $J$ can be constructed in parallel --- each subsystem $i$ performs Subroutine \ref{alg:construct_local_matrix} to construct $J_i$.

\begin{algorithm}[ht]
\floatname{algorithm}{Subroutine}
\caption{Local sub-matrix for subsystem $i$} \label{alg:construct_local_matrix}
\begin{algorithmic}[1]
\Statex \textbf{inputs:} $[H]_i$, $[k]_i$, $\epsilon$ \Statex \textbf{output:} $J_i$
\State Compute $w = [H]_i^\dagger [k]_i$
\State \textbf{if} $\|[H]_i w - [k]_i \|_\infty > \epsilon$ \textbf{:}
\Statex \quad $J_i \leftarrow \textbf{False}$
\Statex \textbf{else :}
\Statex \quad $J_i \leftarrow I - [H]_i^\dagger [H]_i$
\Statex \textbf{return}
\end{algorithmic}
\end{algorithm}

Subroutine \ref{alg:construct_local_matrix} checks whether the dynamics and locality constraints are feasible by checking \eqref{eq:feasibility_check_local}, and if so, returns the appropriate submatrix $J_i$. Notice that the quantity $[H]_i$ is used in both the feasibility check and in $J_i$. Also, $x_0$ does not appear, as we are using the vector of ones in its place.

Having obtained $J$ corresponding to a given locality constraint, we need to check its rank to verify whether global performance is preserved, i.e. $\rank(J) = N_uT$, as per Theorem \ref{thm:equal_trajsets_2}. As previously described, we restrict ourselves to locality constraints of the $d$-local neighborhood type, preferring smaller values of $d$ as these correspond to lower complexity for the localized MPC algorithm \cite{amoalonso_dlmpc1_journal}. Thus, in Algorithm \ref{alg:main}, we start with the smallest possible value of $d=1$, i.e. subsystems communicate only with their immediate neighbors. If $d=1$ does not preserve global performance, we iteratively increment $d$, construct $J$, and check its rank, until optimal global performance is attained. 

\addtocounter{algorithm}{-1}
\begin{algorithm}[ht]
\caption{Optimal local region size} \label{alg:main}
\begin{algorithmic}[1]
\Statex \textbf{inputs:} $A$, $B$, $T$, $\epsilon$
\Statex \textbf{output:} $d$
\Statex \textbf{for} $d = 1 \ldots N$ \textbf{:}
\State \quad \textbf{for} $i = 1 \ldots N$ \textbf{:}
\Statex \quad \quad Construct $[H]_i$, $[k]_i$
\Statex \quad \quad Run Subroutine \ref{alg:construct_local_matrix} to obtain $J_i$
\Statex \quad \quad \textbf{if} $J_i$ is \textbf{False} \textbf{:}
\Statex \quad \quad \quad \textbf{continue} 
\State \quad Construct $J = \begin{bmatrix} J_1 & \ldots & J_N \end{bmatrix}$
\State \quad \textbf{if} $\rank(J) = N_uT$ \textbf{:}
\Statex \quad \quad \textbf{return} \textit{d}
\end{algorithmic}
\end{algorithm}

In Step 1 of Algorithm \ref{alg:main}, we call Subroutine \ref{alg:construct_local_matrix} to check for feasibility and construct submatrices $J_i$. If infeasibility is encountered, or if optimal global performance is not attained, we increment $d$; otherwise, we return optimal locality size $d$.

\subsection{COMPLEXITY}
To analyze complexity, we first make some simplifying scaling assumptions. Assume that the number of states $N_x$ and inputs $N_u$ are proportional to the number of subsystems $N$, i.e. $O(N_x+N_u) = O(N)$. Also, assume that the number of nonzeros $\nM$ for locality constraint corresponding to parameter $d$ are related to one another as $O(\nM) = O(NdT)$.

Steps 1 and 3 determine the complexity of the algorithm. In Step 1, each subsystem performs operations on matrix $[H]_i$, which has size of approximately $N_x(T+1)$ by $dT$ --- the complexity will vary depending on the underlying implementations of the pseudoinverse and matrix manipulations, but will generally scale according to 
$O((dT)^2N_xT)$, or $O(T^3d^2N)$. In practice, $d$ and $T$ are typically much smaller than $N$, and this step is extremely fast; we show this in the next section. 

In Step 3, we perform a rank computation on a matrix of size $(N_x+N_u)T$ by $\nM$. The complexity of this operation, if Gaussian elimination is used, is $O((N_x+N_u)^{1.38}T^{1.38}\nM)$, or $O(T^{2.38}N^{2.38}d)$. Some speedups can be attained by using techniques from \cite{cheung2013_MatrixRank}, which leverage sparsity --- typically, $J$ is quite sparse, with less than 5\% of its entries being nonzero. In practice, Step 3 is the dominating step in terms of complexity.

We remark that this algorithm needs only to be run once offline for any given localized MPC problem. Given a system and predictive horizon, practitioners should first determine the optimal locality size $d$ using Algorithm \ref{alg:main}, then run the appropriate online algorithm from \cite{amoalonso_dlmpc1_journal}.

\section{SIMULATIONS} \label{sec:simulations}

We first present simulations to supplement runtime characterizations of Algorithm \ref{alg:main} from the previous section. Then, we use the algorithm to investigate how optimal locality size varies depending on system size, actuation density, and prediction horizon length. We find that optimal locality size is primarily a function of actuation density. We also verify in simulation that localized MPC performs identically to global MPC when we use the optimal locality size provided by Algorithm \ref{alg:main}, as expected. Code needed to replicate all simulations can be found at \url{https://github.com/flyingpeach/LocalizedMPCPerformance}. This code makes use of the SLS-MATLAB toolbox \cite{li2019_SLSMatlab}, which includes an implementation of Algorithm \ref{alg:main}.

\textit{Note on MATLAB implementation}: When implementing subroutine \ref{alg:construct_local_matrix} in MATLAB, use of the backslash operator (i.e. \texttt{H\textbackslash k}) is faster than the standard pseudoinverse function (i.e. \texttt{pinv(H)*k}). The backslash operator also produces $J_i$ matrices that are as sparse as possible, which facilitates faster subsequent computations.

\subsection{SYSTEM AND PARAMETERS}
We start with a two-dimensional $n$ by $n$ square mesh. Every pair of neighboring nodes have a 40\% probability to be connected by an edge. The expected number of edges is $0.8 \times n \times (n-1)$. We consider only connected graphs\footnote{\changed{Disconnected system graphs imply some amount of dynamical decoupling, which makes distributed control easier than the connected case. We consider only the harder case in simulations --- naturally, our method will do just as well on the easier case.}}. \changed{An example is given in Fig. \ref{fig:grid_topology}. The underlying assumption of this random procedure is that each node has some physical location, and edges are only formed between nodes that are physically nearby --- many other real-world systems (e.g. transport systems, water canals) exhibit similar connectivity and sparsity patterns.} Each node $i$ represents a two-state subsystem of linearized and discretized swing equations
\begin{subequations}
\begin{equation}
\begin{bmatrix} \theta(t+1) \\ \omega (t+1) \end{bmatrix}_i = \sum_{j \in \Ncal_1 (i)} [A]_{ij} \begin{bmatrix} \theta(t) \\ \omega (t) \end{bmatrix}_j + [B]_i [u]_i
\end{equation}
\begin{equation}
\begin{aligned}
\relax [A]_{ii} & = \begin{bmatrix}
    1 & \Delta t \\ -\frac{k_i}{m_i}\Delta t & 1 - \frac{d_i}{m_i}\Delta t
\end{bmatrix}, \\
[A]_{ij} & = \begin{bmatrix}
    0 & 0 \\ \frac{k_{ij}}{m_i} \Delta t & 0
\end{bmatrix}, 
[B]_i = \begin{bmatrix}
    0 \\ 1
\end{bmatrix}
\end{aligned}
\end{equation}
\end{subequations}
where $[\theta]_i$, $[\omega]_i$, and $[u]_i$ are the phase angle deviation, frequency deviation, and control action of the controllable load of bus $i$. Parameters are $m_i^{-1}$ (inertia), $d_i$ (damping), and $k_{ij}$ (coupling); they are randomly drawn from uniform distributions over $[0, 2]$, $[1, 1.5]$ and $[0.5, 1]$ respectively. Self term $k_i$ is defined as $k_i := \sum_{j \in \Ncal_1(i)} k_{ij}$. Discretization step size $\Delta t$ is $0.2$.

\begin{figure}
	\centering
	\includegraphics[width=0.5\columnwidth]{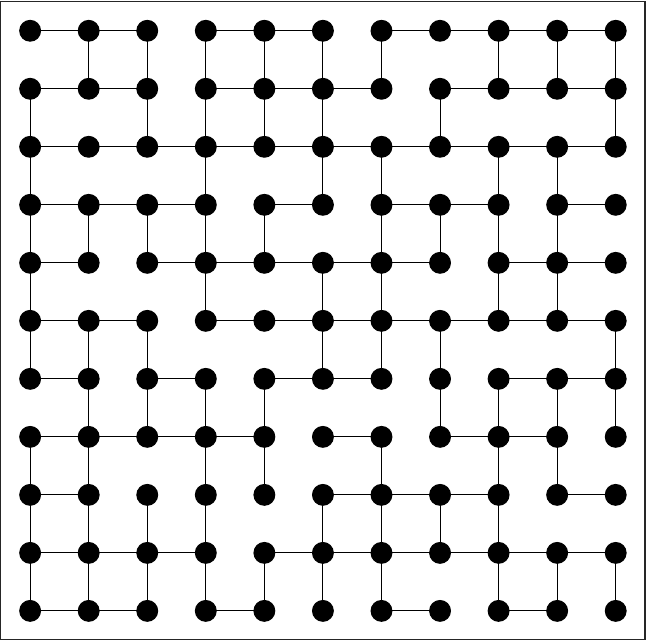}
	\caption{\changed{An example of an 11 by 11 grid topology generated using the described random procedure. Each node on the grid is associated with two states (angle and frequency).}}
	\label{fig:grid_topology}
\end{figure}

Under the given parameter ranges, the system is typically neutrally stable, with a spectral radius of 1. The baseline system parameters are $n=5$ (corresponding to a $5 \times 5$ grid containing 25 nodes, or 50 states) and 100\% actuation. Note: this does not correspond to `full actuation' in the standard sense; it means that each subsystem (which contains 2 states) has one actuator --- 50\% of states are actuated. We use a prediction horizon length of $T=15$ unless otherwise stated.

\subsection{ALGORITHM RUNTIME}

We plot the runtime of Algorithm \ref{alg:main} in Fig. \ref{fig:algorithm_runtime} for different system sizes and horizon lengths. We separately consider runtimes for matrix construction (Step 1) and rank determination (Step 3). The former is parallelized, while the latter is not. 

Runtime for matrix construction is extremely small. Even for the grid with 121 subsystems (242 states), this step takes less than a millisecond. Interestingly, matrix construction runtime also stays relatively constant with increasing network size, despite the worst-case runtime scaling linearly with $N$, as described in the previous section. This is likely due to the sparse structure of $H$. Conversely, matrix construction runtime increases with increasing horizon length.

Runtime for rank determination dominates total algorithm runtime, and increases with both system size and horizon length. Rank determination runtime appears to increase more sharply with increasing horizon length than with increasing system size. Further runtime reductions may be achieved by taking advantage of techniques described in the previous section; however, even without additional speedups, the runtime is no more than 10 seconds for the grid with 242 states.

\begin{figure}
	\centering
	\includegraphics[width=\columnwidth]{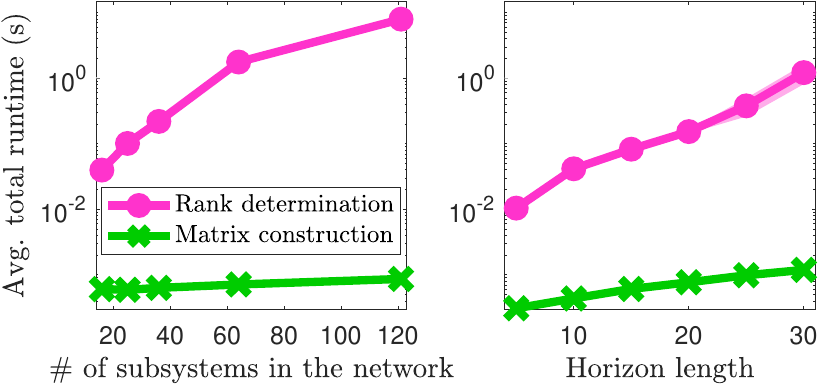}
	\caption{Runtime of matrix construction (Step 1, green) and rank determination (Step 3, pink) of Algorithm \ref{alg:main} vs. network size and horizon length. Parallelized (i.e. per-subsystem) runtimes are shown for matrix construction. The algorithm was run for grids containing 16, 25, 36, 64, and 121 subsystems. For each point, we run the algorithm on five different systems, and plot the average and standard deviation --- here, the standard deviation is so small that it is barely visible. As expected, the rank determination step dominates total runtime, while the matrix construction step is extremely fast.}
	\label{fig:algorithm_runtime}
\end{figure}

\subsection{OPTIMAL LOCALITY SIZE AS A FUNCTION OF SYSTEM PARAMETERS}

We characterize how optimal locality size changes as a function of the system size and horizon length --- the results are summarized in Fig. \ref{fig:parameter_effects}. Actuation density is the main factor that affects optimal locality size. Remarkably, at 100\% actuation, the optimal locality size always appears to be $d=1$, the smallest possible size (i.e. communication only occurs between nodes that share an edge). As we decrease actuation density, the required optimal locality size increases. This makes sense, as unactuated nodes must communicate to at least the nearest actuated node, and the distance to the nearest actuated node grows as actuation density decreases.

The optimal locality size also increases as a function of system size --- but only when we do not have 100\% actuation. At 60\% actuation, for 121 nodes, the optimal locality size is around $d=5$; this still corresponds to much less communication than global MPC.

Predictive horizon length does not substantially impact optimal locality size. At short horizon lengths ($T \leq 10$), we see some small correlation, but otherwise, optimal locality size stays constant with horizon size. Similarly, the stability of the system (i.e. spectral radius) appears to not affect optimal locality size; this was confirmed with simulations over systems with spectral radius of 0.5, 1.0, 1.5, 2.0, and 2.5 for 60\%, 80\%, and 100\% actuation (not included in the plots).

\begin{figure}
	\centering
	\includegraphics[width=\columnwidth]{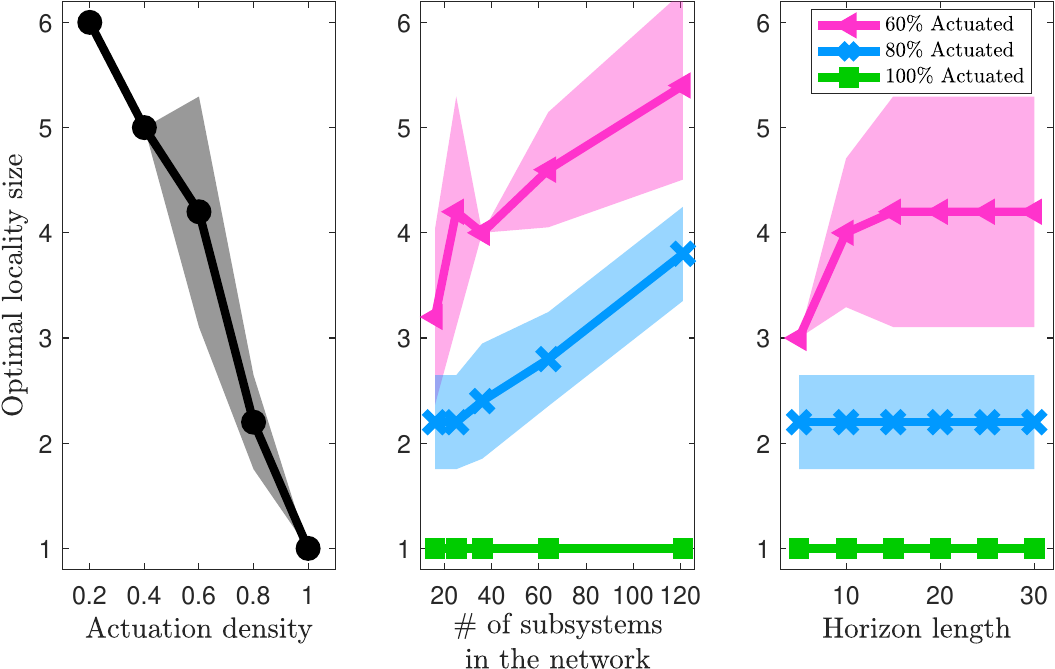}	
	\caption{Optimal locality size as a function of various parameters. Each point represents the average over five different systems; standard deviations are shown by the fill area. \textit{\textbf{(Left)}} Optimal locality size vs. actuation density. The two are inversely correlated. \textit{\textbf{(Center)}} Optimal locality size vs. network size for 60\% actuation (pink), 80\% actuation (blue), and 100\% actuation (green). For 60\% and 80\% actuation, optimal locality size roughly increases with network size. For 100\% actuation, the optimal locality size is always 1, independent of network size. \textit{\textbf{(Right)}} Optimal locality size vs. predictive horizon length for 60\% actuation (pink), 80\% actuation (blue), and 100\% actuation (green). For 60\% and 80\% actuation, optimal locality size increases with horizon size up until $T=10$, then stays constant afterward. For 100\% actuation, the optimal locality size is always 1.}
	\label{fig:parameter_effects}
\end{figure}

\subsection{LOCALIZED PERFORMANCE}

From the previous section, we found that with 100\% actuation, the optimal locality size is always 1. This means that even in systems with 121 subsystems, each node need only communicate with its immediate neighbors (i.e. 4 or less other nodes) to attain optimal global performance. This is somewhat surprising, as this is a drastic (roughly 30-fold) reduction in communication compared to global MPC. To confirm this result, we ran simulations on 20 different systems of size $N=121$ and 100\% actuation. We use LQR objectives with random positive diagonal matrices $Q$ and $R$, and state bounds $\theta_i \in [-4, 4]$ for phase states and $\omega_i \in [-20, 20]$ for frequency states. We use random initial conditions where each value $(x_0)_i$ is drawn from a uniform distribution over $[-2, 2]$.

For each system, we run localized MPC with $d=1$, then global MPC, and compare their costs over a simulation of 20 timesteps. Over 20 simulations, we find the maximum cost difference between localized and global MPC to be 5.6e-6. Thus, we confirm that the reported optimal locality size is accurate, since the cost of localized and global MPC are nearly identical. \changed{As previously described, localized MPC enjoys additional scalability benefits. In this particular case, the size of each subproblem is very small, since it involves at most 10 states\footnote{\changed{Each subproblem consists of up to 5 subsystems: the self subsystems plus a maximum of 4 neighbors. Each subsystem has 2 states}}. This is the size of the subproblem for \textit{any} system size --- if there are 10 nodes, each node will solve a 10-state subproblem in parallel; if there are 100, or even 1000 nodes, each node will still solve a 10-state subproblem in parallel, and the resulting runtime and complexity will be the same \cite{amoalonso_dlmpc1_journal}. This is in contrast with standard (i.e. global) MPC, in which the complexity will increase with system size.}

\subsection{FEASIBILITY VS. OPTIMALITY OF LOCALITY CONSTRAINTS}

In Algorithm \ref{alg:main}, a given locality size is determined to be suboptimal if (1) the locality constraints are infeasible, or (2) the locality constraints are feasible, but matrix $J$ has insufficient rank. The example from Section \ref{sec:numerical_example} suggests that the second case is rare --- to further investigate, we performed 200 random simulations, in which all parameters were randomly selected from uniform distributions --- grid size $N$ from $[4, 11]$ (corresponding to system sizes of up to 121 subsystems), actuation density from $[0.2, 1.0]$, spectral radius from $[0.5, 2.5]$ and horizon length from $[3, 20]$. In these 200 simulations, we encountered 4 instances where a locality constraint was feasible but resulted in insufficient rank; in the vast majority of cases, if a locality constraint was feasible, the rank condition was satisfied as well. 

\section{CONCLUSIONS} \label{sec:conclusions}

In this work, we provided analysis and guarantees on locality constraints and global performance. We presented lemmas, theorems, and an algorithm to certify optimal global performance --- these are the first results of their kind, to the best of our knowledge. We then leveraged these theoretical results to provide an algorithm that determines the optimal locality constraints that will expedite computation while preserving the performance --- this is the first exact method to compute the optimal locality parameter $d$ for DLMPC.

Several directions of future work may be explored:
\begin{enumerate}
    \item The results in this work can be leveraged to investigate the relationship between network topology and optimal locality constraints, i.e. the strictest communication constraints that still preserve optimal global performance. Certain topologies may require long-distance communication between handful of nodes; others may require no long-distance communications. A more thorough characterization will help us understand the properties of systems that are suited for localized MPC.
    \item Algorithm \ref{alg:main} considers $d$-local communication constraints; a natural extension is to consider non-uniform local communication constraints, which are supported by the theory presented in this work. A key challenge of this research direction is the combinatorial nature of available local communication configurations; insights from the research direction suggested above will likely help narrow down said set of configurations.
    \item Simulations suggest that feasibility of a given locality constraint overwhelmingly coincides with optimal global performance. This poses the question of whether feasibility of a given locality constraint is \textit{sufficent} for $J$ to be full rank under certain conditions, and what these conditions may be. Additional investigation could reveal more efficient implementations of Algorithm \ref{alg:main}, as bypassing the rank checking step would save a substantial amount of computation time.
    \item This work focuses on nominal trajectories. Additional investigation is required to characterize the impact of locality constraints on trajectories robust to disturbances. For polytopic disturbances, the space (or minima) of available values of $\mathbf{\Xi}g$ from (10) in \cite{amoalonso_dlmpc1_journal} is of interest. Due to the additional variables in the robust MPC problem, we cannot directly reuse techniques from this paper, though similar ideas may be applicable. 
\end{enumerate}
\bibliography{refs}
\bibliographystyle{IEEEtran}

\end{document}